\documentclass{pspum-l}
\usepackage{graphics}
\usepackage{color}
\usepackage[bookmarksopen,bookmarksdepth=2]{hyperref}
\usepackage{amsfonts,amsmath,amssymb}
\usepackage{enumitem} 
\usepackage{xfrac,slashed,bbm,mathtools,ulem,epigraph}	
\usepackage{upgreek}
\usepackage{threeparttable}
\usepackage{tikz}	
\usepackage{pict2e}	
\usepackage{eso-pic} 
\usepackage{caption}
\usepackage{subcaption}
\usepackage{tikz-feynman} 

\newcommand{\N}{\mathbb N}

\newcommand{\R}{\mathbb R}
\newcommand{\C}{\mathbb C}

\renewcommand{\H}{\mathcal H}

\newcommand{\A}{\mathcal{A}}

\newcommand{\E}{\mathcal{E}}

\newcommand{\mB}{\mathcal{B}}

\newcommand{\Esg}{\mathcal{E}_s^{\gamma}}

\def\Re{{\mathrm{Re}}}
\def\Im{{\mathrm{Im}}}

\newcommand{\bra}[1]{\left\langle #1 \right|}
\newcommand{\ket}[1]{\left| #1 \right\rangle}

\newcommand{\p}[2]{\langle #1 , #2 \rangle}

\newcommand{\spn}{\textnormal{span}}

\renewcommand{\mid}{~\middle|~}

\newcommand{\Tr}{\operatorname{Tr}}

\def\cs{\mathrm{cs}}
\newcommand{\sa}{\textnormal{sa}}
\newcommand{\br}[1]{\langle#1\rangle}
\newcommand{\bbr}[1]{\llangle #1 \rrangle^{1L}_N}
\newcommand{\brr}[1]{\boldsymbol\prec\!#1\!\boldsymbol\succ}
\newcommand{\ev}{\textup{ev}}
\newcommand{\odd}{\textup{odd}}

\renewcommand{\th}{^{\text{th}}}
\newcommand{\image}[2]{\includegraphics[width=#1\textwidth, height=#1\textheight,keepaspectratio]{#2}}

\usepackage{simpler-wick}
\usepackage{soul} 

\makeatletter
\newsavebox{\@brx}
\newcommand{\llangle}[1][]{\savebox{\@brx}{\(\m@th{#1\langle}\)}%
  \mathopen{\copy\@brx\kern-0.5\wd\@brx\usebox{\@brx}}}
\newcommand{\rrangle}[1][]{\savebox{\@brx}{\(\m@th{#1\rangle}\)}%
  \mathclose{\copy\@brx\kern-0.5\wd\@brx\usebox{\@brx}}}
\makeatother

\tikzset{edge/.style={decorate, decoration=snake}}
\tikzset{puntjes/.style={dash, pattern=on 0pt off 2\pgflinewidth}}
\tikzset{streepjes/.style={dash pattern=on 8pt off 4pt}}
\newcommand{\ncvertex}[1]{\filldraw (#1) circle (12.5pt);
	{\color{white}\filldraw (#1) circle (10.2pt);}
	\node at (#1) {\scalebox{0.9}{\huge $f$}};}

\usepackage{amsthm}
\newtheorem{thm}{Theorem}[section]
\newtheorem{lem}[thm]{Lemma}
\newtheorem{prop}[thm]{Proposition}
\newtheorem{cor}[thm]{Corollary}
\newtheorem{defi}[thm]{Definition}

\newtheorem{exam}[thm]{Example}

\begin{document}

\title{Cyclic cocycles and one-loop corrections in the spectral action}
\author{Teun D.H. van Nuland}
\address{Radboud University, Heyendaalseweg 135, 6525 AJ Nijmegen, The Netherlands}
\email{t.vannuland@math.ru.nl}

\author{Walter D. van Suijlekom}
\address{Radboud University, Heyendaalseweg 135, 6525 AJ Nijmegen, The Netherlands}
\email{waltervs@math.ru.nl}

\date{\today}

\begin{abstract}
We present an intelligible review of recent results concerning cyclic cocycles in the spectral action and one-loop quantization.
We show that the spectral action, when perturbed by a gauge potential, can be written as a series of Chern--Simons
actions and Yang--Mills actions of all orders. In the odd orders, generalized Chern--Simons forms are integrated against an
odd $(b,B)$-cocycle, whereas, in the even orders, powers of the curvature are integrated against $(b,B)$-cocycles that are
Hochschild cocycles as well. In both cases, the Hochschild cochains are derived from the Taylor series expansion of the
spectral action $\Tr(f(D+V))$ in powers of $V=\pi_D(A)$, but unlike the Taylor expansion we expand in increasing order
of the forms in $A$. We then analyze the perturbative quantization of the spectral action in noncommutative geometry
and establish its one-loop renormalizability as a gauge theory. We show that the one-loop counterterms are of the same
Chern--Simons--Yang--Mills form so that they can be safely subtracted from the spectral action. A crucial role will be played
by the appropriate Ward identities, allowing for a fully spectral formulation of the quantum theory at one loop.
  
\end{abstract}

\maketitle

\tableofcontents

\section{Introduction}
\label{sct:Cyclic Introduction}
The spectral action \cite{CC96,CC97} is one of the key instruments in the applications of noncommutative geometry to particle physics. With inner fluctuations \cite{C96} of a noncommutative manifold playing the role of gauge potentials, the spectral action principle yields the corresponding Lagrangians. Indeed, the asymptotic behavior of the spectral action for small momenta leads to experimentally testable field theories, by interpreting the spectral action as a classical action and applying the usual renormalization group techniques. In particular, this provides the simplest way known to geometrically explain the dynamics and interactions of the gauge bosons and the Higgs boson in the Standard Model Lagrangian as an effective field theory \cite{CCM07} (see also the textbooks \cite{CM07,Sui14}). More general noncommutative manifolds (spectral triples) can also be captured by the spectral action principle, leading to models beyond the standard model as well. As shown in \cite{CC06}, if one restricts to the scale-invariant part, one may naturally identify a Yang--Mills term and a Chern--Simons term to elegantly appear in the spectral action. From the perspective of quantum field theory, the appearance of these field-theoretic action functionals sparks hope that we might find a way to go beyond the classical framework provided by the spectral action principle. It is thus a natural question whether we can also field-theoretically describe the full spectral action, without resorting to the scale-invariant part. 


Motivated by this, we study the spectral action when it is expanded in terms of inner fluctuations associated to an arbitrary noncommutative manifold, without resorting to heat-kernel techniques. Indeed, the latter are not always available and an understanding of the full spectral action could provide deeper insight into how gauge theories originate from noncommutative geometry. Let us now give a more precise description of our setup.

We let $(\A,\H,D)$ be an finitely summable spectral triple. If $f : \R \to \C$ is a suitably nice function we may define the spectral action \cite{CC97}:
$$
\Tr (f(D)).
$$
An inner fluctuation, as explained in \cite{C96}, is given by a Hermitian universal one-form
\begin{align}\label{eq:A uitgeschreven}
	A=\sum_{j=1}^n a_jdb_j\in\Omega^1(\A),
\end{align}
for elements $a_j,b_j\in\A$. The terminology `fluctuation' comes from representing $A$ on $\H$ as
\begin{align}\label{eq:V uitgeschreven}
	V:=\pi_D(A)=\sum_{j=1}^n a_j[D,b_j]\in\mB(\H)_\sa,
\end{align}
and fluctuating $D$ to $D+V$ in the spectral action.
 %
The variation of the spectral action under the inner fluctuation is then given by
\begin{align}\label{variation of SA}
	\Tr(f(D+V))-\Tr(f(D)).
\end{align}
As spectral triples can be understood as noncommutative spin$^\text{c}$ manifolds (see \cite{C08}) encoding the gauge fields as an inner structure, one could hope that perturbations of the spectral action could be understood in terms of noncommutative versions of geometrical, gauge theoretical concepts. Hence we would like to express \eqref{variation of SA} in terms of universal forms constructed from $A$. To express an action functional in terms of universal forms, one is naturally led to cyclic cohomology. As it turns out, hidden inside the spectral action we will identify an odd $(b,B)$-cocycle $(\tilde\psi_1,\tilde\psi_3,\ldots)$ and an even $(b,B)$-cocycle $(\phi_2,\phi_4,\ldots)$ for which $b\phi_{2k}=B\phi_{2k}=0$, i.e., each Hochschild cochain $\phi_{2k}$ forms its own $(b,B)$-cocycle $(0,\ldots,0,\phi_{2k},0,\ldots)$. On the other hand, the odd $(b,B)$-cocycle $(\tilde\psi_{2k+1})$ is truly infinite (in the sense of \cite{C94}). 

 The key result is that for suitable $f:\R\to\C$ we may expand
\begin{align}\label{eq:expansion intro}
	\Tr(f(D+V)-f(D))=\sum_{k=1}^\infty\left(\int_{\psi_{2k-1}}\cs_{2k-1}(A)+\frac{1}{2k}\int_{\phi_{2k}}F^{k}\right),
\end{align}
	in which the series converges absolutely. Here $\psi_{2k-1}$ is a scalar multiple of $\tilde\psi_{2k-1}$, $F_t=tdA+t^2A^2$, so that $F=F_1$ is the curvature of $A$, and $\cs_{2k-1}(A)=\int_0^1 AF_t^{k-1}dt$ is a generalized noncommutative Chern--Simons form. 

As already mentioned, a similar result was shown earlier to hold for the scale-invariant part $\zeta_D(0)$ of the spectral action. Indeed,  Connes and Chamseddine \cite{CC06} expressed the variation of the scale-invariant part in dimension $\leq 4 $ as
\begin{equation*}
\zeta_{D+V}(0) - \zeta_D(0) = - \frac 1 4\int_{\tau_0} (dA+A^2) + \frac 12 \int_\psi \left(A d A + \frac 2 3 A^3\right),
\end{equation*}
for a certain Hochschild 4-cocycle $\tau_0$ and cyclic 3-cocycle $\psi$.

It became clear in \cite{NS21} that an extension of this result to the full spectral action is best done by using multiple operator integrals \cite{ST} instead of residues. It allows for stronger analytical results, and in particular allows to go beyond dimension $4$.
Moreover, for our analysis of the cocycle structure that appears in the full spectral action we take the Taylor series expansion as a starting point, and for working with such expansions multiple operator integrals provide the ideal tools, as shown by the strong results in \cite{ACDS09,CS18,Skr13,Sui11}. In \cite{NS21} we pushed these results further still, by proving estimates and continuity properties for the multiple operator integral when the self-adjoint operator has an $s$-summable resolvent, thereby supplying the discussion here with a strong functional analytic foundation.
This article will start with a review of the results of \cite{NS21} without involving multiple operator integration techniques. Through the use of abstract brackets, we will investigate the interesting cyclic structure that exists within the spectral action, with all analytical details taking place under the hood.


We work out two interesting possibilities for application of our main result and the techniques used to obtain it. The first application is to index theory. One can show that the $(b,B)$-cocycles $\phi$ and $\psi$ are \textit{entire} in the sense of \cite{C88a}. This makes it meaningful to analyze their pairing with K-theory, which we find to be trivial in Section \ref{sct:vanishing pairing}. 

The second application is to quantization. In Section \ref{sct:One-Loop}, though evading analytical difficulties, we will take a first step towards the quantization of the spectral action within the framework of spectral triples. Using the asymptotic expansion proved in Theorem \ref{thm:asymptotic expansion}, and some basic quantum field theoretic techniques, we will propose a one-loop quantum effective spectral action and show that it satisfies a similar expansion formula, featuring in particular a new pair of $(b,B)$-cocycles.

Although the main aim of this paper is to give a simple review of the results of \cite{NS21} and \cite{NS21b}, some essential novelty is also provided.
 In order to connect to the quantization results of \cite{NS21b}, the results of \cite{NS21} are slightly generalized as well as put into context.
Moreover, this paper gives a mathematically precise underpinning of the results presented in \cite{NS21b}, which was geared towards a physics audience. We hope that the discussion presented here is clear to mathematicians with or without affinity to physics.

\section{Taylor expansion of the spectral action}
Consider a finitely summable spectral triple $(\A,\H,D)$ (in the sense that for some $s$ the operator $(i-D)^{-s}$ is trace-class). Given the fluctuations of $D$ to $D+V$ as explained in the introduction, we are interested in a Taylor expansion of the spectral action: 
	\begin{align}\label{eq:spectraction brackets}
		\Tr(f(D+V)-f(D)) &=\sum_{n=1}^\infty\frac{1}{n!}\frac{d^n}{dt^n}\Tr(f(D+tV))\big|_{t=0}\nonumber\\
		&=\sum_{n=1}^\infty\frac{1}{n}\br{V,\ldots,V},
	\end{align}
where $\br{V,\ldots,V}$ is a notation for ($1/(n-1)!$ times) the $n^{\text{th}}$ derivative of the spectral action, defined below, and dependent on $f$ and $D$.
        Such an expansion exist under varying assumptions on $f$, $D$, and $V$, see for instance \cite{Han06,Sui11,Skr13,NSkr21,NS21}. When we are interested in the inner fluctuations of the form $V = \pi_D(A)$ as in Equation \eqref{eq:V uitgeschreven}, a convenient function class in which $f$ should lie is given as in \cite{NS21} by
\begin{align}\label{eq:function class}
	\Esg:=\left\{f\in C^\infty \mid
	\begin{aligned}
	&\text{there exists $C_f\geq 1$ s.t. } \|\widehat{(fu^m)^{(n)}}\|_1\leq (C_f)^{n+1}n!^\gamma \\
	&\text{for all $m=0,\ldots,s$ and $n\in\mathbb N_0$}
	\end{aligned}
	\right\},
\end{align}
for $\gamma\in (0,1]$ a number, and $s$ the summability of the pertinent spectral triple.   Indeed, as shown in \cite{NS21} if $f \in \Esg$ we have good control over the expansion appearing on the right-hand side of \eqref{eq:spectraction brackets}.

  For our present expository purposes, however, it is sufficient to assume that $f'$ is compactly supported and analytic in a region of $\C$ containing a rectifiable curve $\Gamma$ which surrounds the support of $f$ in $\R$. In this case we have
  \begin{align}
    \label{eq:br cycl}
		\br{V_1,\ldots,V_n}= \frac{1}{2\pi i}\oint_\Gamma f'(z) \Tr\left(\prod_{j=1}^n V_j(z-D)^{-1}\right).
  \end{align}
  A concrete expression can be also obtained in terms of divided differences of $f$. Indeed, for a self-adjoint operator $D$ in $\H$ with compact resolvent, we let $\varphi_1,\varphi_2,\ldots$ be an orthonormal basis of eigenvectors of $D$, with corresponding eigenvalues $\lambda_1,\lambda_2,\ldots$. Recall Cauchy's integral formula for divided differences \cite[Chapter I.1]{Don74}:
  $$
g[x_0, \ldots x_n] = \frac{1}{2\pi i} \oint \frac{g(z) }{(z-x_0)\cdots (z-x_n)} dz,
$$
with the contour enclosing the points $x_i$. This then yields 
\begin{align}
\br{V,\ldots,V}	&=\frac{1}{n}\sum_{i_1,\ldots,i_n\in\N}f'[\lambda_{i_1},\ldots,\lambda_{i_n}]V_{i_1i_2}\cdots V_{i_{n-1}i_n}V_{i_{n}i_1}.\label{eq:SA divdiff}
\end{align}
where $V_{kl}:=\p{\varphi_k}{V\varphi_l}$ denote the matrix elements of $V$.
This formula appears in \cite[Corollary 3.6]{Han06} and, in higher generality, in \cite[Theorem 18]{Sui11}.
The formula \eqref{eq:SA divdiff} gives a very concrete way to calculate derivatives of the spectral action, as well as to calculate the Taylor series of a perturbation of the spectral action.

For our algebraic results we only need two simple properties of the bracket $\br{\cdot}$, stated in the following lemma.

\begin{lem}\label{cycl bracket}
For $V_1,\ldots,V_n\in\mB(\H)$ and $a\in\A$ we have
\begin{enumerate}[label=\textnormal{(\Roman*)}]
	\item $\br{V_1,\ldots,V_n}=\br{V_n,V_1,\ldots,V_{n-1}},$\label{cyclicity}
	\item $ \br{aV_1,V_2,\ldots,V_n}-\br{V_1,\ldots,V_{n-1},V_{n}a}=\br{V_1,\ldots,V_{n},[D,a]}$.\label{commutation}
\end{enumerate}
\end{lem}

\begin{proof}
  We will omit all analytical details and give a proof for finite-dimensional Hilbert spaces only. The full proof involving multiple operator integrals can be found in \cite{NS21} (as Lemma 14).

  In finite-dimensions we may use formula \eqref{eq:br cycl} for the bracket. Clearly (I) then follows directly from the tracial property. Note that the left-hand side of equality (II) comes down to the commutator of $a$ with the resolvent $(z-D)^{-1}$, for which we have the equality
  $$
(z-D)^{-1} a - a(z-D)^{-1} = (z-D)^{-1} [D,a] (z-D)^{-1} 
  $$
This readily leads to the right-hand side in (II).
\end{proof}

\section{Cyclic cocycles in the spectral action}
We now generalize a little and consider a collection of functions $\brr{\cdot} : \mB(\H)^{\times n}\to\R$, $n\in\N$, satisfing
	\begin{enumerate}[label=\textnormal{(\Roman*)}]
	\item $\brr{V_1,\ldots,V_n}~=~\brr{V_n,V_1,\ldots,V_{n-1}},$\label{cyclicity general}
	\item $\brr{aV_1,V_2,\ldots,V_n}-\brr{V_1,\ldots,V_{n-1},V_na}~=~\brr{V_1,\ldots,V_n,[D,a]}$\label{commutation general}
\end{enumerate}		
In view of Lemma \ref{cycl bracket} above, the brackets $\br{\cdot}$ that appear in the Taylor expansion of the spectral action form a special case of these generalized brackets $\brr{\cdot}$ ---and of course form the key motivation for introducing them. However, such structures pop up in other places as well, for instance \cite{Liu22,GSW}, cf. \cite[Proposition 3.2 and Remark 3.2]{Hock}. In Section \ref{sct:One-Loop}, we will introduce yet another instance of $\brr{\cdot}$, in order to obtain one-loop corrections.

Therefore, in contrast to \cite{NS21}, the following discussion will involve the abstract bracket $\brr{\cdot}$ instead of the explicit $\br{\cdot}$.

\subsection{Hochschild and cyclic cocycles}
When the above brackets $\brr{\cdot}$ are evaluated at one-forms $a[D,b]$ associated to a spectral triple, the relations (I) and (II) can be translated nicely in terms of the coboundary operators appearing in cyclic cohomology. This is very similar to the structure appearing in the context of index theory, see for instance \cite{GS89,Hig06}.

Let us start by recalling the definition of Hochschild cochains and the boundary operators $b$ and $B$ from \cite{C85}.

\begin{defi}
If $\A$ is an algebra, and $n\in\N_0$, we define the space of \textit{Hochschild $n$-cochains}, denoted by $\mathcal{C}^n(\A)$, as the space of $(n+1)$-linear
functionals $\phi$ on $\A$ with the property that if $a_j =1$ for some $j \geq 1$, then $\phi(a_0,\ldots,a_n) = 0$.
\end{defi}
For such cochains we may use, as in \cite{C94}, an integral notation on universal differential forms that is defined by linear extension of 
	$$\int_{\phi}a_0da_1\cdots da_n:=  \phi(a_0,a_1,\ldots,a_n).$$

\begin{defi}
Define operators $b : \mathcal{C}^{n}(\A) \to \mathcal{C}^{n+1}(\A)$ and $B: \mathcal{C}^{n+1}(\A) \to \mathcal{C}^{n}(\A)$ by
\begin{align*}
b\phi(a_0, a_1,\dots, a_{n+1})
:=& \sum_{j=0}^n (-1)^j \phi(a_0,\dots, a_j a_{j+1},\dots, a_{n+1})\\
& + (-1)^{n+1} \phi(a_{n+1} a_0, a_1,\dots, a_n) ,\\
B \phi(a_0 ,a_1, \ldots, a_n) :=& 
\sum_{j=0}^n (-1)^{nj}\phi(1,a_j,a_{j+1},\ldots, a_{j-1}).
\end{align*}
\end{defi}
Note that $B = \mathbf{A} B_0$ in terms of the operator $ \mathbf{A}$ of cyclic anti-symmetrization and the operator defined by $B_0 \phi (a_0, a_1, \ldots, a_n) = \phi(1,a_0, a_1,\ldots, a_n)$. Note that in integral notation we simply have
$$
\int_{B_0 \phi} a_0 d a_1 \cdots d a_n = \int_{\phi} da_0 da_1 \cdots da_n. 
$$
One may check that the pair $(b,B)$ defines a double complex, \textit{i.e.} $b^2 = 0,~ B^2=0,$ and $bB +Bb =0$. Hochschild cohomology now arises as the cohomology of the complex  $(\mathcal{C}^n(\A),b)$. In contrast, we will be using \textit{periodic cyclic cohomology}, which is defined as the cohomology of the totalization of the $(b,B)$-complex. That is to say, 
\begin{align*}
\mathcal{C}^\ev(\A) = \bigoplus_k \mathcal{C}^{2k} (\A) ; \qquad \mathcal{C}^{\odd}(\A) = \bigoplus_k \mathcal{C}^{2k+1} (\A),
\end{align*}
form a complex with differential $b+B$ and the cohomology of this complex is called periodic cyclic cohomology. We will also refer to a periodic cyclic cocycle as a cyclic cocycle or a $(b,B)$-cocycle. Explicitly, an odd $(b,B)$-cocycle is thus given by a sequence
$$
(\phi_1, \phi_3, \phi_5, \ldots),
$$
where $\phi_{2k+1} \in \mathcal{C}^{2k+1}(\A)$ and 
$$
b \phi_{2k+1} + B \phi_{2k+3} = 0 ,
$$
for all $k \geq 0$, and also $B \phi_1 = 0$. An analogous statement holds for even $(b,B)$-cocycles.

\subsection{Cyclic cocycles associated to the brackets}
\label{sct:Cyclic cocycles associated to multiple operator integrals}

In terms of the generic bracket $\brr{\cdot}$ satisfying \ref{cyclicity general} and \ref{commutation general}, we define the following Hochschild $n$-cochain:
\begin{align}\label{eq:def phi_n}
	\phi_n(a_0,\ldots,a_n):=\brr{a_0[D,a_1],[D,a_2],\ldots,[D,a_{n}]} \qquad  (a_0, \ldots, a_n \in \A).
\end{align}
	We easily see that $B_0\phi_n$ is invariant under cyclic permutations, so that $B\phi_n=nB_0\phi_n$ for odd $n$ and $B\phi_n=0$ for even $n$. Also, $\phi_n(a_0,\ldots,a_n)=0$ when $a_j=1$ for some $j\geq1$. We put $\phi_0:=0$.
\begin{lem}\label{lem:b}
	We have $b\phi_n=\phi_{n+1}$ for odd $n$ and we have $b\phi_n=0$ for even $n$.
\end{lem}
\begin{proof}
  We only consider the case $n=1$ while referring to \cite[Lemma 17]{NS21} for the proof of the general case. We combine the definition of the $b$-operator with Leibniz' rule for $[D,\cdot]$ to obtain: 
   \begin{align*}
     \int_{b \phi_1} a_0 d  a_1 d  a_2
     & = \brr{a_0 a_1 [D,a_2] } - \brr{a_0 [D,a_1 a_2]} + \brr{a_2 a_0 [D,a_1]}\\
           &=
         -  \brr{a_0 [D,a_1 ]a_2} + \brr{a_2 a_0 [D,a_1]} = \brr{a_0 [D,a_1],[D,a_2]}
   \end{align*}
   where we used (II) for the last equality.
  \end{proof}
\begin{lem}\label{lem:c}
Let $n$ be even. We have $bB_0\phi_n=2\phi_n-B_0\phi_{n+1}$.
\end{lem}
\begin{proof}
Again we only consider the first case $n=2$ while referring to \cite[Lemma 17]{NS21} for the proof of the general case
\begin{align*}
&           \int_{b B_0 \phi_2} a_0 d  a_1 d  a_2  = \int_{B_0 \phi_2} a_0 a_1 d  a_2 - \int_{B_0 \phi_2} a_0 d (a_1  a_2) +\int_{B_0 \phi_2} a_2 a_0 d  a_1   \\
           &\qquad= \brr{[D,a_0a_1],[D,a_2]} - \brr{[D,a_0],[D,a_1a_2]} + \brr{[D,a_2 a_0],[D,a_1]}\\
           & \qquad = \cdots = 2 \brr{a_0[D,a_1],[D,a_2]} - \brr{[D,a_0],[D,a_1],[D,a_2]},
\end{align*}
combining Leibniz'  rule with (I) and (II). 
\end{proof}

Motivated by these results we define 
\begin{equation}
  \label{eq:psi}
  \psi_{2k-1}:=\phi_{2k-1}-\tfrac{1}{2}B_0\phi_{2k},
  \end{equation}
so that
$$B\psi_{2k+1}=2(2k+1)b\psi_{2k-1}.$$
We can rephrase this property in terms of the $(b,B)$-complex as follows. 
\begin{prop}
  \label{prop:bB}
  Let $\phi_n$ and $\psi_{2k-1}$ be as defined above and set
  $$\tilde{\psi}_{2k-1}:=(-1)^{k-1}\frac{(k-1)!}{(2k-1)!}\psi_{2k-1}\,.$$
  \begin{enumerate}[label=\textnormal{(\roman*)}]
  \item The sequence $(\phi_{2k})$ is a $(b,B)$-cocycle and each $\phi_{2k}$ defines an even Hochschild cocycle: $b \phi_{2k} = 0$. 
    \item The sequence $(\tilde \psi_{2k-1})$ is an odd $(b,B)$-cocycle. 
    \end{enumerate}
  \end{prop}

\subsection{The brackets as noncommutative integrals}

We will now describe how brackets $\brr{V,\ldots, V}$ can be written as noncommutative integrals of certain universal differential forms defined in terms of $A = \sum a_j db_j\in \Omega^1(\A)$, using only property (I) and (II).

At first order not much exciting happens and we simply have
$$
\brr{V} =\sum_j  \brr{a_j[D,b_j]}= \sum_j\int_{\phi_1}  a_j db_j = \int_{\phi_1} A.
$$
More interestingly, at second order we find using property (II) of the bracket that
\begin{align*}
  \brr{V,V } &=\sum_{j,k} \brr{ a_j[D,b_j] , a_k[D,b_k] } \\
  &=\sum_{j,k} \brr{ a_j[D,b_j] a_k,[D,b_k] } + \sum_{j,k} \brr{ a_j[D,b_j] ,[D,a_k],[D,b_k] }\\
  &= \int_{\phi_2} A^ 2 + \int_{\phi_3} A d A.
\end{align*}
Continuing like this, while only using property (II) of the bracket we find
\begin{align*}
	\brr{V,V,V}&=\int_{\phi_3}A^3+\int_{\phi_4}AdAA+\int_{\phi_5}AdAdA,\\
	\brr{V,V,V,V}&=\int_{\phi_4}A^4+\int_{\phi_5}(A^3dA+AdAA^2)+\int_{\phi_6}AdAdAA+\int_{\phi_7}AdAdAdA.
\end{align*}
This implies that, at least when the infinite sum on the left-hand side makes sense:
	\begin{align*}
	\sum_n \frac 1 n \brr{V,\ldots, V} =&\int_{\phi_1} A+\frac{1}{2}\int_{\phi_2}A^2+\int_{\phi_3}\Big(\frac{1}{2}AdA+\frac{1}{3}A^3\Big)\\
	&+\int_{\phi_4}\Big(\frac{1}{3}AdAA+\frac{1}{4}A^4\Big)+\ldots,
	\end{align*}
	where the dots indicate terms of degree 5 and higher. Using $\phi_{2k-1}=\psi_{2k-1}+\frac{1}{2}B_0\phi_{2k}$, this becomes
	\begin{align*}
		\sum_n \frac 1 n \brr{V,\ldots, V}=&\int_{\psi_1} A+\frac{1}{2}\int_{\phi_2}(A^2+dA)+\int_{\psi_3}\Big(\frac{1}{2}AdA+\frac{1}{3}A^3\Big)\\
		&+\frac{1}{4}\int_{\phi_4}\Big(dAdA+\frac{2}{3}(dAA^2+AdAA+A^2dA)+A^4\Big)+\ldots.
	\end{align*}
	Notice that, if $\phi_4$ would be tracial, we would be able to identify the terms $dAA^2$, $AdAA$ and $A^2dA$, and thus obtain the Yang--Mills form $F^2=(dA+A^2)^2$, under the fourth integral. In the general case, however, cyclic permutations under $\int_\phi$ produce correction terms, of which one needs to keep track. 
        Indeed, using \cite[Corollary 24]{NS21} we may re-order the integrands to yield
\begin{align*}
	\sum_n \frac 1 n \brr{V,\ldots, V}=&\int_{\psi_1}A+\int_{\phi_2}\tfrac{1}{2}(dA+A^2)+\int_{\psi_3}(\tfrac{1}{2}dAA+\tfrac{1}{3}A^3)+\tfrac{1}{4}\int_{\phi_4}(dA+A^2)^2\nonumber\\
	&+\int_{\psi_5}(\tfrac{1}{3}(dA)^2A+\tfrac{1}{2}dAA^3+\tfrac{1}{5}A^5)+\tfrac{1}{6}\int_{\phi_6}(dA+A^2)^3+\ldots
\end{align*}
where the dots indicate terms of degree 7 and higher. Writing $F=dA+A^2$ and $\cs_1(A):=A$, $\cs_3(A):=\tfrac12 dAA+\tfrac13 A^3$, etc., we can already discern our desired result in low orders.

\bigskip

As a preparation for the general result, we briefly recall from \cite{Qui90} the definition of Chern--Simons forms of arbitrary degree. 

         \begin{defi}
           \label{defi:cs}
           The (universal) \textbf{Chern--Simons form} of degree $2k-1$ is given for $A \in \Omega^1(\A)$ by 
           \begin{equation}\label{eq:cs}
\cs_{2k-1}(A) := \int_0^1 A (F_t)^{k-1} \,dt,
           \end{equation}
           where $F_t = t dA + t^ 2 A^2$ is the curvature two-form of the (connection) one-form $A_t = t A$.
           \end{defi}

         \begin{exam}
           For the first three Chern--Simons forms one easily derives the following explicit expressions:
           \begin{gather*}
             \cs_1(A) = A; \qquad  \cs_3(A) = \frac 12 \left( A dA + \frac 2 3 A^3 \right);\\
             \cs_5(A) = \frac 13 \left( A (dA)^2 + \frac 3 4 A dA A^ 2 + \frac 3 4 A^3 dA + \frac 3 5 A^5 \right).
             \end{gather*}
           \end{exam}

\subsection{Cyclic cocycles in the Taylor expansion of the spectral action }
We now apply the above results to the brackets appearing in the Taylor expansion of the spectral action:
$$
\Tr(f(D+V)-f(D)) = \sum_{n=1}^\infty \frac 1 n \br{V, \cdots, V}.
$$
In order to control the full Taylor expansion of the spectral action we naturally need a growth condition on the derivatives of the function $f$,  and this is accomplished by considering the class $\E_s^\gamma$ defined in \eqref{eq:function class}.
The following result is \cite[Theorem 27]{NS21}.
        
\begin{thm}\label{thm:main thm}
  Let $(\A,\H,D)$ be an $s$-summable spectral triple, and let $f\in\E_s^{\gamma}$ for $\gamma\in(0,1)$. The spectral action fluctuated by $V=\pi_D(A)\in\Omega_D^1(\A)_\sa$ 
  can be written as
 \begin{align*}
  \Tr(f(D+V)-f(D)) = \sum_{k=1}^\infty \left( \int_{\psi_{2k-1}}  \cs_{2k-1} (A) +\frac 1 {2k} \int_{\phi_{2k}}  F^{k} \right),
  \end{align*}
  where the series converges absolutely.
\end{thm}

Under less restrictive conditions on the function $f$ we also have the following asymptotic version of this result \cite[Proposition 28]{NS21} 

\begin{thm}\label{thm:asymptotic expansion}
	Let $(\A,\H,D)$ be a spectral triple, and let $\brr{\cdot}$ satisfy \ref{cyclicity} and \ref{commutation}, with associated cyclic cocycles $\phi$ and $\tilde\psi$. For $A\in\Omega^1(\A)$ and $V=\pi_D(A)$, we asymptotically have
	$$\sum_n \frac 1 n \brr{V,\ldots, V}\sim\sum_{k=1}^\infty\bigg(\int_{\psi_{2k-1}}\cs_{2k-1}(A)+\frac{1}{2k}\int_{\phi_{2k}} F^{k}\bigg),$$
	by which we mean that, for every $K\in\N$, there exist forms $\omega_l\in\Omega^l(\A)$ for $l=K+1,\ldots,2K+1$ such that
	\begin{align*}
		\sum_{n=1}^K\frac{1}{n}\!\brr{V,\ldots,V}&-\sum_{k=1}^K\left(\int_{\psi_{2k-1}}\cs_{2k-1}(A)+\frac{1}{2k}\int_{\phi_{2k}}F^k\right)
		=\sum_{l=K+1}^{2K+1}\int_{\phi_l}\omega_l.
	\end{align*}
\end{thm}	
	In particular, by taking $\brr{\cdot}=\br{\cdot}$, we obtain the following corollary.
\begin{cor}	
	For $f\in C^\infty$, and $V=\pi_D(A)\in\Omega^1_D(\A)_\sa$ such that the Taylor expansion of the spectral action converges, we asymptotically have
	\begin{align*}
		\Tr(f(D+V)-f(D))&=\sum_{n=1}^\infty \frac{1}{n!}\frac{d^n}{dt^n}\Tr(f(D+tV))\Big|_{t=0}\\
		&\sim\sum_{k=1}^\infty\bigg(\int_{\psi_{2k-1}}\cs_{2k-1}(A)+\frac{1}{2k}\int_{\phi_{2k}} F^{k}\bigg).
	\end{align*}	
\end{cor}

\subsection{Gauge invariance and the pairing with K-theory}\label{sct:vanishing pairing}

 Since the spectral action is a spectral invariant, it is in particular invariant under conjugation of $D$ by a unitary $U\in \A$. More generally, in the presence of an inner fluctuation we find that the spectral action is invariant under the transformation
$$
D+V \mapsto U (D+V) U^* = D + V^U; \qquad V^U = U[D,U^*] + U V U^*.
$$
This transformation also holds at the level of the universal forms, with a gauge transformation of the form $A \mapsto A^U = U d U^* + U A U ^*$. Let us analyze the behavior of the Chern--Simons and Yang--Mills terms appearing in Theorem \ref{thm:main thm} under this gauge transformation, and derive an interesting consequence for the pairing between the odd $(b,B)$-cocycle $\tilde \psi$ with the odd K-theory group of $\A$. As an easy consequence of the fact that $\phi_{2k}$ is a Hochschild cocycle, we have 
\begin{lem}
  The Yang--Mills terms $\int_{\phi_{2k}} F^k$ with $F = dA +A^2$ are invariant under the gauge transformation $A \mapsto A^U$ for every $k \geq 1$. 
  \end{lem}

We are thus led to the conclusion that the  sum of Chern--Simons forms  is gauge invariant as well. Indeed, arguing as in \cite{CC06}, since both $\Tr (f(D+V))$ and the Yang--Mills terms are invariant under $V \mapsto V^U$, we find that, under the assumptions stated in Theorem \ref{thm:main thm}:
$$
\sum_{k=0}^\infty \int_{\psi_{2k+1}} \cs_{2k+1}(A^U ) =  
\sum_{k=0}^\infty \int_{\psi_{2k+1}} \cs_{2k+1} (A ).
$$
Each individual Chern--Simons form behaves non-trivially under a gauge transformation. Nevertheless, it turns out that we can conclude, just as in \cite{CC06}, that the pairing of the whole $(b,B)$-cocycle with K-theory is trivial. 
Since the $(b,B)$-cocycle $\tilde \psi$ is given as an infinite sequence, we should first carefully study the analytical behavior of $\tilde \psi$. In fact, we should show that it is an \textit{entire cyclic cocycle} in the sense of \cite{C88a} (see also \cite[Section IV.7.$\alpha$]{C94}). 
It turns out \cite[Lemma 36]{NS21} that our assumptions on the growth of the derivatives of $f$ ensure that the brackets define entire cyclic cocycles. 
\begin{lem}
	Fix $f\in\Esg$ for $\gamma<1$ and equip $\A$ with the norm $\| a \|_1 = \| a \| + \| [D,a]\|$. Then, for any bounded subset $\Sigma\subset\mathcal{A}$ there exists $C_\Sigma$ such that
		$$\left |\tilde{\psi}_{2k+1}(a_0,\ldots, a_{2k+1}) \right|\leq \frac{C_\Sigma}{k!},$$
	for all $a_j\in\Sigma$. Hence, $\phi$ and $\tilde\psi$ are entire cyclic cocycles.
\end{lem}
We thus have the following interesting consequence of Theorem \ref {thm:main thm}.

\begin{thm}
  Let $f\in\Esg$ for $\gamma<1$. Then the pairing of the odd entire cyclic cocycle $\tilde \psi$ with $K_1(\A)$ is trivial, \textit{i.e.}
  $$\langle U,\tilde\psi\rangle=(2\pi i)^{-1/2}\sum_{k=0}^\infty (-1)^{k}k!\tilde{\psi}_{2k+1}(U^*,U,\ldots,U^*,U)=0
  $$
  for all unitary $U\in \A$.
\end{thm}

\section{One-loop corrections to the spectral action}
\label{sct:One-Loop}

We now formulate a quantum version of the spectral action. To do this, we must first interpret the spectral action, expanded in terms of generalized Chern--Simons and Yang--Mills actions by Theorem \ref{thm:main thm}, as a classical action, which leads us naturally to a noncommutative geometric notion of a vertex. Enhanced with a spectral gauge propagator derived from the formalism of random matrices (and in particular, random finite noncommutative geometries) this gives us a concept of one-loop counterterms and a proposal for a one-loop \textit{quantum effective spectral action}, without leaving the spectral framework. We will show here that, at least in a finite-dimensional setting, these counterterms can again be written as Chern--Simons and Yang--Mills forms integrated over (quantum corrected) cyclic cocycles. We therefore discern a renormalization flow in the space of cyclic cocycles.

\subsection{Conventions}
  
We let $\varphi_1,\varphi_2,\ldots$ be an orthonormal basis of eigenvectors of $D$, with corresponding eigenvalues $\lambda_1,\lambda_2,\ldots$. For any $N\in\N$, we define
\begin{align*}
	H_N:=(M_N)_\sa,\quad M_N:=\spn\left\{\ket{\varphi_i}\bra{\varphi_j}:~i,j\in\{1,\ldots,N\}\right\},
\end{align*}
and endow $H_N$ with the Lebesgue measure on the coordinates $Q\mapsto \Re( Q_{ij})$ ($i\leq j$) and $Q\mapsto\Im(Q_{ij})$ ($i<j$). Here and in the following, $Q_{ij}:=\p{\varphi_i}{Q\varphi_j}$ are the matrix elements of $Q$.
For simplicity, we will assume that the perturbations $V_1,\ldots,V_n$ are in $\cup_K H_K$. 

For us, a \textit{Feynman diagram} is a finite multigraph with a number of marked vertices of degree 1 called external vertices, all other vertices being called internal vertices or, by abuse of terminology, vertices. An edge, sometimes called a propagator, is called external if it connects to an external vertex, and internal otherwise. The external vertices are simply places for the external edges to attach to, and are often left out of the discussion. An $n$-point diagram is a Feynman diagram with $n$ external edges. A Feynman diagram is called one-particle-irreducible if any multigraph obtained by removing one of the internal edges is connected.

\subsection{Diagrammatic expansion of the spectral action}
\label{sect:sa}

Viewing the spectral action as a classical action, and following the background field method, the vertices of degree $n$ in the corresponding quantum theory should correspond to $n\th$-order functional derivatives of the spectral action. However, in the paradigm of noncommutative geometry, a base manifold is absent, and functional derivatives do not exist in the local sense. Therefore, a more abstract notion of a vertex is needed. The brackets $\br{\cdot}$ from \eqref{eq:SA divdiff} that power the expansion of the spectral action in Theorems \ref{thm:main thm}  and \ref{thm:asymptotic expansion} are by construction cyclic and multilinear extensions of the derivatives of the spectral action, and as such provide an appropriate notion of \textit{noncommutative vertices}.
We define a noncommutative vertex with $V_1,\ldots,V_n\in\cup_K H_K$ on the external edges by
\begin{align}
\raisebox{-43pt}{\scalebox{0.45}{
\begin{tikzpicture}[thick]
	\draw[edge] (0,2) to (2,2);
	\draw[edge] (1,0) to (2,2);
	\draw[edge] (1,4) to (2,2);
	\draw[edge] (3,4) to (2,2);
	\draw[edge] (4,2) to (2,2);
	\ncvertex{2,2}
	\draw[line width=2pt, line cap=round, dash pattern=on 0pt off 3\pgflinewidth] (2,0) arc (-90:0:1.5cm);
	\node at (-0.5,2) {\huge $V_1$};
	\node at (0.8,4.4) {\huge $V_2$};
	\node at (3.2,4.4) {\huge $V_3$};
	\node at (4.5,2) {\huge $V_4$};
	\node at (0.7,-0.4) {\huge $V_n$};
\end{tikzpicture}}}
\quad
:=\quad\br{V_1,\ldots,V_n}.
  \label{eq:bracket}
\end{align}

In contrast to a normal vertex of a Feynman diagram, a noncommutative vertex is decorated with a  cyclic order on the edges incident to it. By convention, the edges are attached clockwise with respect to this cyclic order. 
As such, with perturbations $V_1,\ldots,V_n$ decorating the external edges, the diagram \eqref{eq:bracket} reflects the cyclicity of the bracket: $\br{V_1,\ldots,V_n}=\br{V_n,V_1,\ldots,V_{n-1}}$, the first property of Lemma \ref{cycl bracket}. 
In order to diagramatically represent the second property of Lemma \ref{cycl bracket} as well, we introduce the following notation. Wherever a gauge edge meets a noncommutative vertex we can insert a dashed line decorated with an element $a\in\A$ before or after the gauge edge, with the following meaning:
\begin{align*}
\raisebox{-10pt}{\scalebox{0.45}{
\begin{tikzpicture}[thick]
	\draw (0,0) arc (110:70:4cm);
	\draw (0,0.025) arc (110:70:4cm);
	\draw (0,0.05) arc (110:70:4cm);
	\draw (0,0.075) arc (110:70:4cm);
	\draw[edge] (2.65,0) to (2.65+0.5,2);
	\draw[streepjes] (2.65,0) to (2.3,2);
	\node at (2.2,2.5) {\huge $a$};
	\node at (3.3,2.5) {\huge $V$};
\end{tikzpicture}}}
\quad
\raisebox{10pt}{
:=
}
\raisebox{-10pt}{\scalebox{0.45}{
\begin{tikzpicture}[thick]
	\draw (0,0) arc (110:70:4cm);
	\draw (0,0.025) arc (110:70:4cm);
	\draw (0,0.05) arc (110:70:4cm);
	\draw (0,0.075) arc (110:70:4cm);
	\draw[edge] (2.65,0) to (2.65+0.5,2);
	\node at (3.3,2.5) {\huge $aV$};
\end{tikzpicture}}}
\quad
\raisebox{10pt}{
,
}
\qquad\qquad
\raisebox{-10pt}{\scalebox{0.45}{
\begin{tikzpicture}[thick]
	\draw (0,0) arc (110:70:4cm);
	\draw (0,0.025) arc (110:70:4cm);
	\draw (0,0.05) arc (110:70:4cm);
	\draw (0,0.075) arc (110:70:4cm);
	\draw[edge] (0,0) to (-0.5,2);
	\draw[streepjes] (0,0) to (0.35,2);
	\node at (0.3,2.5) {\huge $a$};
	\node at (-0.65,2.5) {\huge $V$};
\end{tikzpicture}}}
\quad
\raisebox{10pt}{
:=
}
\quad
\raisebox{-10pt}{
\scalebox{0.45}{\begin{tikzpicture}[thick]
	\draw (0,0) arc (110:70:4cm);
	\draw (0,0.025) arc (110:70:4cm);
	\draw (0,0.05) arc (110:70:4cm);
	\draw (0,0.075) arc (110:70:4cm);
	\draw[edge] (0,0) to (-0.5,2);
	\node at (-0.65,2.5) {\huge $Va$};\end{tikzpicture}}}.
\end{align*}
With this notation, the equation
\begin{align}\label{eq:classical Ward}
\br{a V_1,\ldots,V_n}-\br{V_1,\ldots,V_n a}
        =\br{V_1,\ldots,V_n,[D,a]},
\end{align}
is represented as
\begin{align}
\label{eq:ward}
\raisebox{-10pt}{\scalebox{0.45}{
\begin{tikzpicture}[thick]
	\draw (0,0) arc (110:70:4cm);
	\draw (0,0.025) arc (110:70:4cm);
	\draw (0,0.05) arc (110:70:4cm);
	\draw (0,0.075) arc (110:70:4cm);
	\draw[streepjes] (2.65,0) to (2.3,2);
	\node at (2.2,2.5) {\huge $a$};
\end{tikzpicture}}}
\quad\,\,\,
-
\,\,\quad
\raisebox{-10pt}{\scalebox{0.45}{
\begin{tikzpicture}[thick]
	\draw (0,0) arc (110:70:4cm);
	\draw (0,0.025) arc (110:70:4cm);
	\draw (0,0.05) arc (110:70:4cm);
	\draw (0,0.075) arc (110:70:4cm);
	\draw[streepjes] (0,0) to (0.35,2);
	\node at (0.3,2.5) {\huge $a$};
\end{tikzpicture}}}
\quad
=
\quad
\raisebox{-10pt}{\scalebox{0.45}{
\begin{tikzpicture}[thick]
	\draw (0,0) arc (110:70:4cm);
	\draw (0,0.025) arc (110:70:4cm);
	\draw (0,0.05) arc (110:70:4cm);
	\draw (0,0.075) arc (110:70:4cm);
	\draw[edge] (1.325,0.3) to (1.325,2);
	\node at (1.325,2.5) {\huge $[D,a]$};
\end{tikzpicture}}}
\quad
,
\end{align}
and is as such referred to as the \textit{Ward identity}.

To illustrate, let us give the relevant lower order computations.
The cyclic cocycles are expressed in terms of diagrams as
\begin{align}
  \label{eq:bracket-cochain}
  \int_{\phi_n} a^0 da^1 \cdots da^n 
  &=
  \raisebox{-41pt}{\scalebox{0.45}{
\begin{tikzpicture}[thick]
	\draw[edge] (0,2) to (2,2);
	\draw[edge] (1,0) to (2,2);
	\draw[edge] (1,4) to (2,2);
	\draw[edge] (3,4) to (2,2);
	\draw[edge] (4,2) to (2,2);
	\ncvertex{2,2}
	\draw[line width=2pt, line cap=round, dash pattern=on 0pt off 3\pgflinewidth] (2.5,-0.5) arc (-85:-15:2cm);
	\node at (-1.4,2) {\huge $a^0[D,a^1]$};
	\node at (0.5,4.4) {\huge $[D,a^2]$};
	\node at (3.5,4.4) {\huge $[D,a^3]$};
	\node at (5.1,2) {\huge $[D,a^4]$};
	\node at (0.7,-0.4) {\huge $[D,a^n]$};
\end{tikzpicture}}}.
\end{align}

\noindent For one external edge we find, writing $A=\sum_j a_jdb_j$ and suppressing summation over $j$,
\begin{align}
\br{V} =   \br{a_j [D,b_j]} &= 
\,\,
\raisebox{-4pt}{\scalebox{.45}{
\begin{tikzpicture}[thick]
	\draw[edge] (1,0) to (3,0);
	\ncvertex{3,0}
	\node at (0.5,0.6) {\huge $a_j[D,b_j]$};
\end{tikzpicture}}}
\quad
= 
\int_{\phi_1} A.
\end{align}
For two external edges, we apply the Ward identity \eqref{eq:ward} and derive
\begin{align*}
  \br{V,V}&=
\quad
\raisebox{-5pt}{\scalebox{0.45}{
\begin{tikzpicture}[thick]
	\draw[edge] (1,0) to (3,0);
	\draw[streepjes] (3.4,0.02) to (5,2);
	\draw[edge] (3,0) to (5,0);
	\ncvertex{3,0}
	\node at (-0.5,0) {\huge $a_j[D,b_j]$};
	\node at (6,0) {\huge $[D,b_{j'}]$};
	\node at (5.6,2.1) {\huge $a_{j'}$};
\end{tikzpicture}}}  
  \\[1mm]
  &=
  \quad
\raisebox{-5pt}{\scalebox{0.45}{
\begin{tikzpicture}[thick]
	\draw[edge] (1,0) to (3,0);
	\draw[streepjes] (2.55,0.08) to (1,2);
	\draw[edge] (3,0) to (5,0);
	\ncvertex{3,0}
	\node at (-0.5,0) {\huge $a_j[D,b_j]$};
	\node at (6,0) {\huge $[D,b_{j'}]$};
	\node at (0.4,2.1) {\huge $a_{j'}$};
\end{tikzpicture}}}
\quad
+
\quad
\raisebox{-5pt}{\scalebox{0.45}{
\begin{tikzpicture}[thick]
	\draw[edge] (1,0) to (3,0);
	\draw[edge] (3,0) to (3,1.6);
	\draw[edge] (3,0) to (5,0);
	\ncvertex{3,0}
	\node at (-0.5,0) {\huge $a_j[D,b_j]$};
	\node at (6,0) {\huge $[D,b_{j'}]$};
	\node at (3,2.1) {\huge $[D,a_{j'}]$};
\end{tikzpicture}}}
\\[4mm]
  &= \int_{\phi_2} A^2 
  + \int_{\phi_3} A d A.
\end{align*}

\subsubsection{The propagator}           
An important part of the quantization process introduced here is to find a mathematical formulation for the propagator. In other words, we need to introduce more general diagrams than the one-vertex diagram in \eqref{eq:bracket}, and assign each an amplitude. As usual in quantum field theory, the amplitudes depend on a cutoff $N$ and are possibly divergent as $N\to\infty$. 

What we will call a \textit{noncommutative Feynman diagram} (or, for brevity, a diagram) is a Feynman diagram in which every internal vertex $v$ is decorated with a cyclic order on the edges incident to $v$. These decorated vertices are what we call the noncommutative vertices, and are denoted as in \eqref{eq:bracket}. 
The edges of a diagram are always drawn as wavy lines. They are sometimes called gauge edges to distinguish them from any dashed lines in the diagram, which do not represent physical particles, but are simply notation. The \textit{loop order} is defined to be $L:=1-V+E$, where $V$ is the amount of (noncommutative) vertices and $E$ is the amount of internal edges. We also say the noncommutative Feynman diagram is $L$-loop, e.g., the noncommutative Feynman diagram in \eqref{eq:bracket} is zero-loop. When the respective multigraph is planar, $L$ corresponds to the number of internal faces. Following physics terminology, these faces are referred to as \textit{loops}.
As usual for Feynman diagrams, the external edges are marked, say by the numbers $1,\ldots,n$.

Note that, by our definition, a noncommutative Feynman diagram is almost the same as a ribbon graph, the sole difference being that ribbons are sensitive to twisting, whereas our edges are not.

Each nontrivial noncommutative Feynman diagram will be assigned an \textit{amplitude}, as follows. Here \textit{nontrivial} means that every connected component contains at least one vertex with nonzero degree.

\begin{defi}\label{def:Propagator}
Let $N\in\N$ and let $f\in C^\infty$ satisfy $f'[\lambda_i,\lambda_j]>0$ for $i,j\leq N$. Given a nontrivial $n$-point noncommutative Feynman diagram $G$ with external vertices marked by $1,\ldots,n$, its \textbf{amplitude} at level $N\in\N$ on the gauge fields $V_1,\ldots,V_n\in\cup_K H_K$ is denoted $\Gamma_N^G(V_1,\ldots,V_n)$, and is defined recursively as follows. When $G$ has precisely one vertex and the markings $1,\ldots,n$ respect its cyclic order, we set $\Gamma_N^G(V_1,\ldots,V_n):=\br{V_1,\ldots,V_n}$. Suppose the amplitudes of diagrams $G_1$ and $G_2$ with external edges $1,\ldots,n$ and $n+1,\ldots,m$ are defined. Then to the disjoint union $G$ of the diagrams we assign the amplitude
\begin{align*}
	\Gamma_N^{G}(V_1,\ldots,V_m):=\Gamma_N^{G_1}(V_1,\ldots,V_n)\Gamma_N^{G_2}(V_{n+1},\ldots,V_m).
\end{align*}
 Suppose the amplitude of a diagram $G$ is defined. Then, for any two distinct numbers $i,j\in\{1,\ldots,n\}$, let $G'$ be the diagram obtained from $G$ by connecting the two external edges $i$ and $j$ by a gauge edge (a propagator). We then define the amplitude of $G'$ as
\begin{align*}
	\Gamma_N^{G'}(V_1,\ldots,\widehat{V_i},\ldots,\widehat{V_j},\ldots,V_{n}):=-\frac{\int_{H_N}\Gamma_N^G(V_1,\ldots,\overset{i}{Q},\ldots,\overset{j}{Q},\ldots,V_{n})e^{-\tfrac12\br{Q,Q}}dQ}{\int_{H_N}e^{-\tfrac12\br{Q,Q}}dQ}.
\end{align*}
\end{defi}

Well-definedness is a straightforward consequence of Fubini's theorem. Note that, in general, $\Gamma^G_N$ is not cyclic in its arguments, as was the case in \eqref{eq:bracket}.

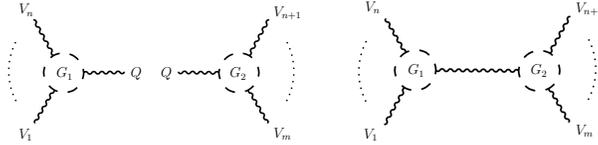
\begin{figure}[h!]
\hspace{30pt}
\scalebox{0.45}{
\hspace{24pt}
\begin{subfigure}[t]{0.4\textwidth}
	\image{0.92}{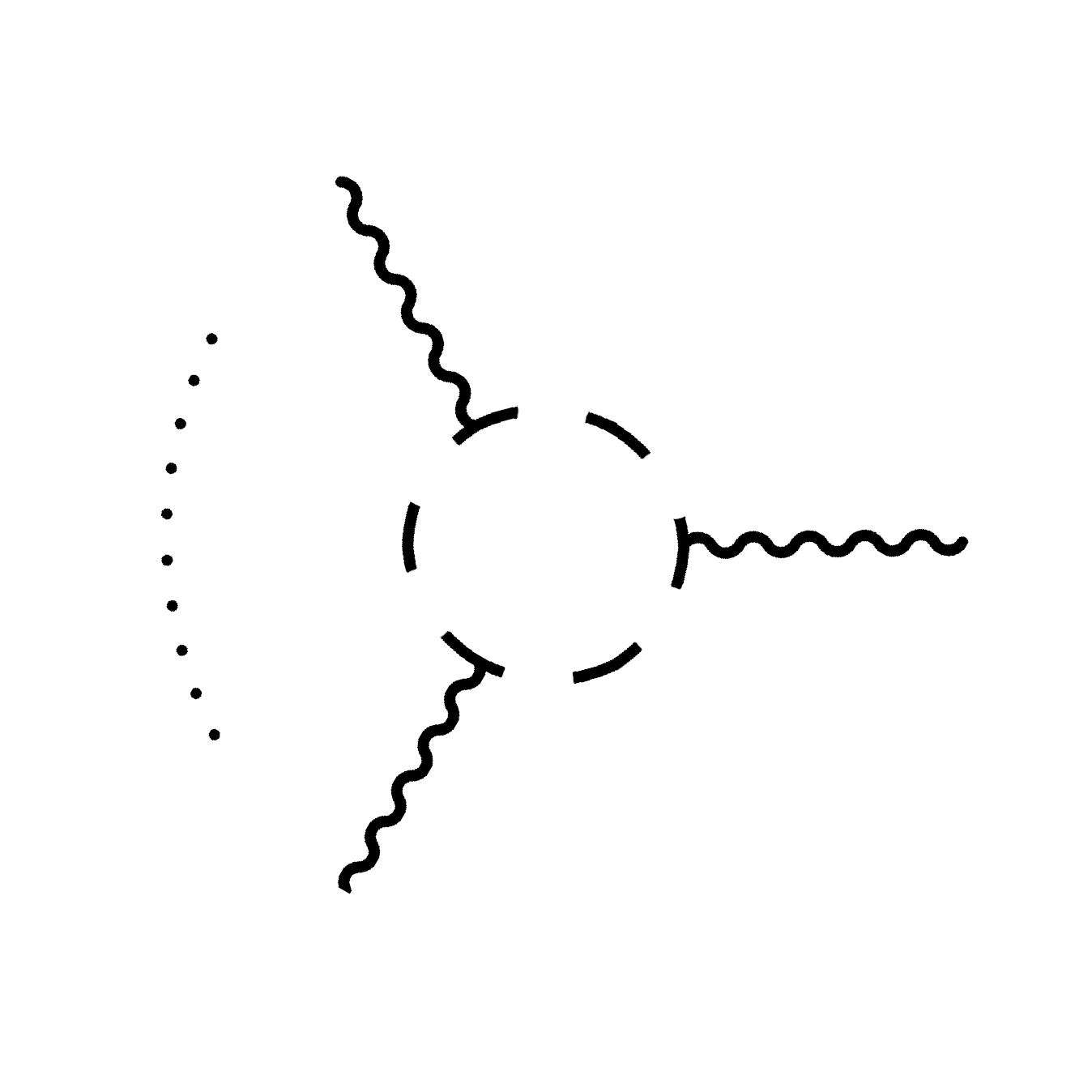}
	\put(-104,116){\Large $V_n$}
	\put(-104,11){\Large $V_1$}
	\put(-10,63){\Large $Q$}
	\put(-72.4,62.5){\Large $G_1$}
\end{subfigure}
\begin{subfigure}[t]{0.4\textwidth}
	\image{0.92}{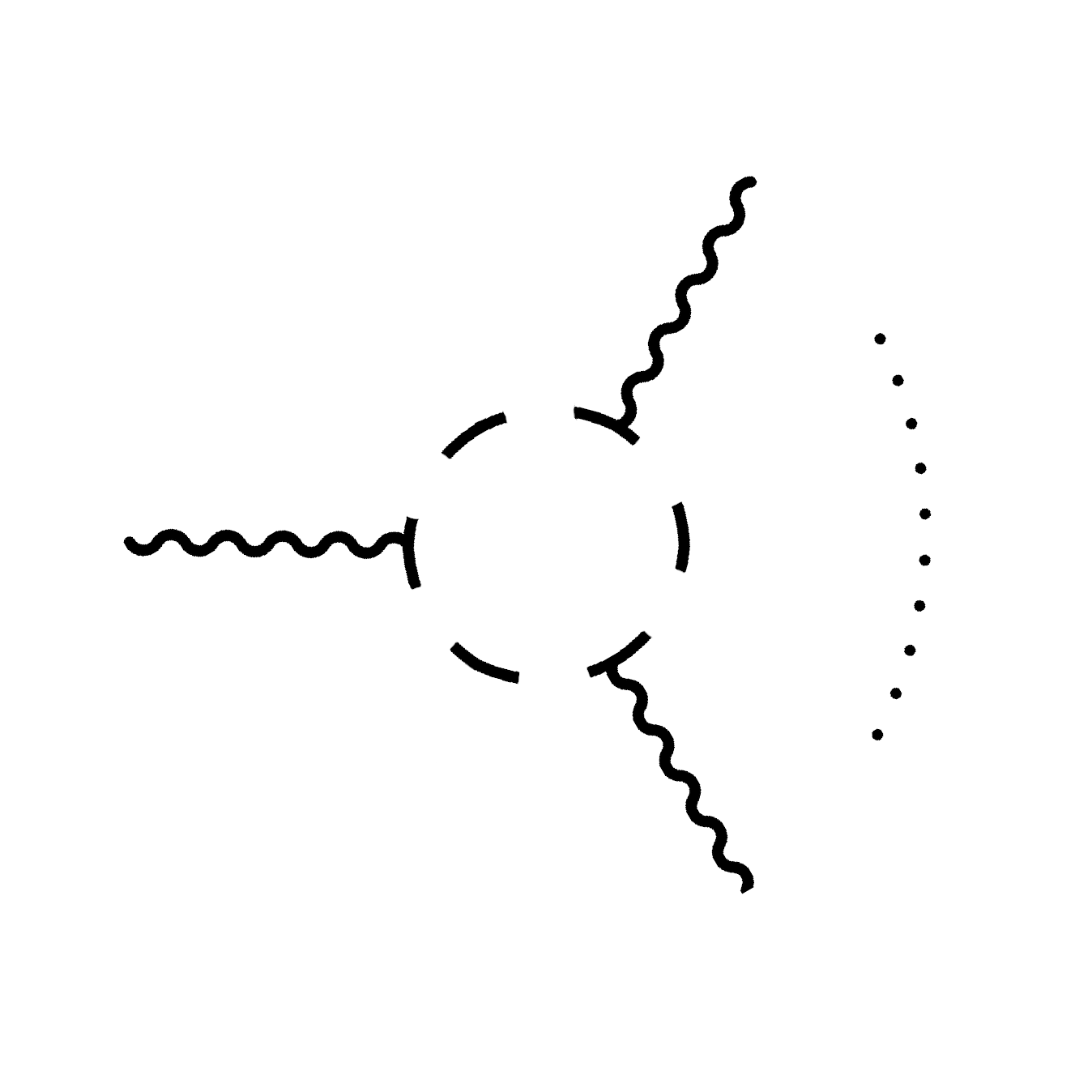}
	\put(-37,114){\Large $V_{n+1}$}
	\put(-37,12){\Large $V_m$}
	\put(-132,63){\Large $Q$}
	\put(-74,62.5){\Large $G_2$}
\end{subfigure}
\begin{subfigure}[t]{\textwidth}
	\image{0.664}{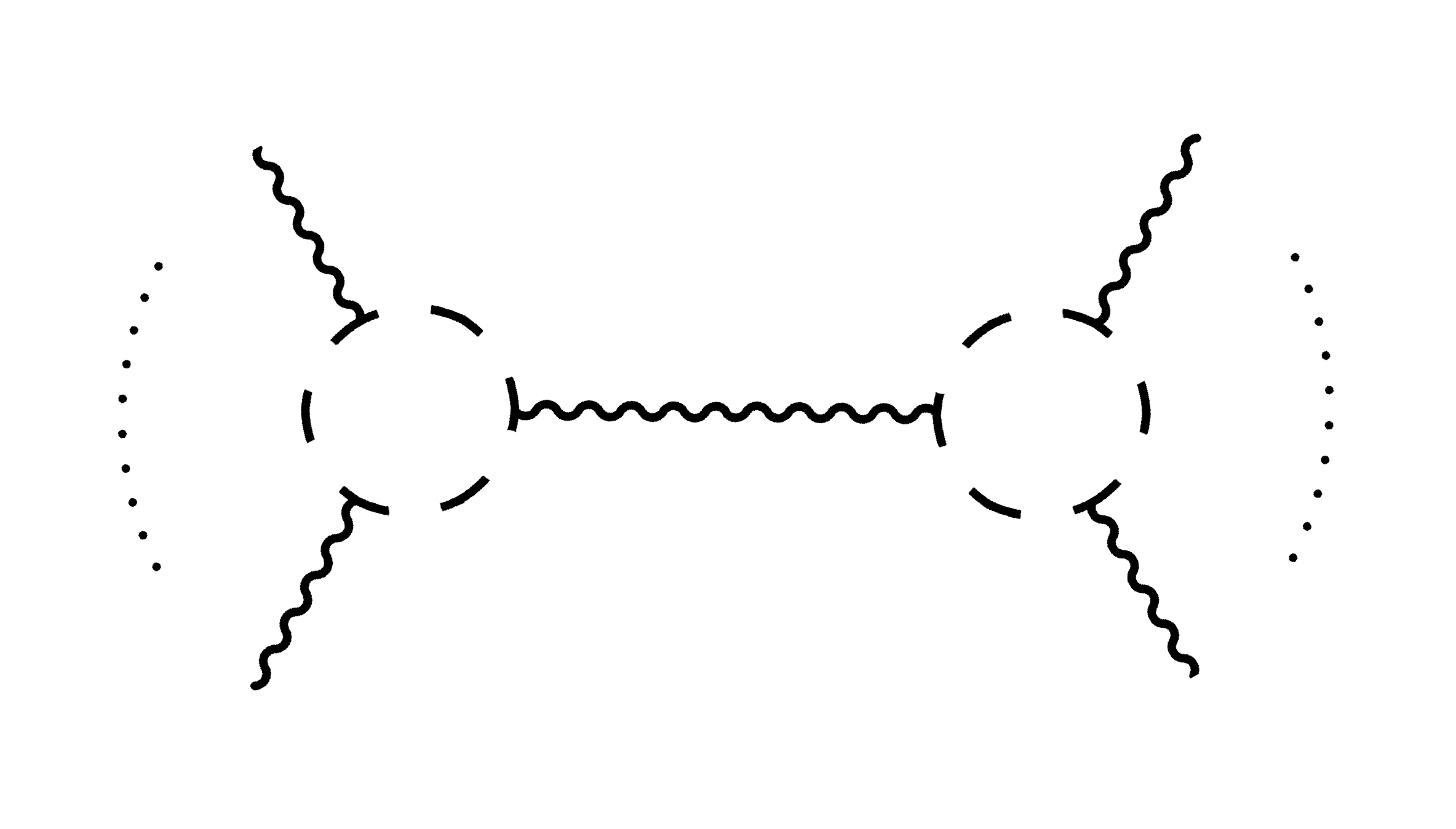} 
	\put(-214,118){\Large $V_n$}
	\put(-215,11){\Large $V_1$}
	\put(-178,65){\Large $G_1$}
	\put(-37,117){\Large $V_{n+1}$}
	\put(-37,14){\Large $V_m$}
	\put(-75,65){\Large $G_2$}
\end{subfigure}
}
\caption{Constructing the propagator.}
\end{figure}

The assumption that $f'[\lambda_i,\lambda_j]>0$ for $i,j\leq N$ can be accomplished by allowing $f$ to be unbounded, and replacing the spectral action
$$\Tr(f(D))$$
with the regularized version
$$\Tr(f_N(D))$$
where $f_N:=f\Phi_N$ for a sequence of bump functions $\Phi_N$ ($N\in\N$) that are 1 on $\{\lambda_k:~k\leq N\}$. As quantization takes place on the finite level (for a finite $N$), it is natural to also regularize the classical action before we quantize.
 Because we can now easily require
 	$$f_N'[\lambda_k,\lambda_l]=f'[\lambda_k,\lambda_l]>0,$$ for all $k,l\leq N$, Definition \ref{def:Propagator} makes sense and can be studied by Gaussian integration as in  \cite[Section 2]{BIZ80}.
 	
\subsection{Loop corrections to the spectral action}
 	
	To obtain the propagator, we have chosen the approach of random noncommutative geometries (as done in \cite{AK19,KP20}, see \cite{BG16,GS19a} for computer simulations) in the sense that the integrated space in Definition \ref{def:Propagator} is the whole of $H_N$. Other approaches are conceivable by replacing $H_N$ by a subspace of gauge fields particular to the gauge theory under consideration (like $\Omega^1_D(\A)_{\sa}$ for a finite spectral triple $(\A,\H,D)$) but this should also take into account gauge fixing, and will quickly become very involved. We expect to require sophisticated machinery to perform such an integration, similar to the machinery in \cite{EF}.

 	In our case, the propagator becomes quite simple, and can be explicitly expressed by the following result.

\begin{lem}\label{lem:propagator}
Let $f\in C^\infty$ satisfy $f'[\lambda_k,\lambda_l]>0$ for $k,l\leq N$. For $k,l,m,n\in\{1,\ldots,N\}$, we have
$$
\frac{ \int_{H_N} Q_{kl}  Q_{mn} e^{-\tfrac12\br{Q,Q}} dQ} {\int_{H_N} e^{-\tfrac12\br{Q,Q} } dQ} = 
\delta_{kn}\delta_{lm} G_{kl},
$$
in terms of $G_{kl} := \frac{1}{f'[\lambda_k, \lambda_l]}$.
\end{lem}
\begin{proof}
By \eqref{eq:SA divdiff} we have the finite sum
\begin{align*}
\br{Q,Q}=\sum_{k,l} f'[\lambda_k,\lambda_l]\left((\Re (Q_{kl}))^2+(\Im(Q_{kl}))^2\right),
\end{align*}
for all $Q\in H_N$. Moreover, we have
\begin{align*}
&\int_{H_N} Q_{kl}  Q_{mn} e^{-\frac 12 \br{Q,Q}} dQ \\
&\quad=  \int_{H_N} (\Re(Q_{kl})\Re(Q_{mn})-\Im(Q_{kl})\Im(Q_{mn}))e^{-\frac 12 \br{Q,Q}} dQ\\
&\qquad+i\int_{H_N}(\Re(Q_{kl})\Im(Q_{mn}) + \Im(Q_{kl})\Re(Q_{mn}))e^{-\frac 12 \br{Q,Q}} dQ.
\end{align*}
The second integral on the right-hand side vanishes because its integrand is an odd function in at least one of the coordinates of $H_N$. The same holds for the first integral whenever $\{k,l\}\neq\{m,n\}$. Otherwise, we use that $\Re(Q_{lk})=\Re(Q_{kl})$ and $\Im(Q_{lk})=-\Im(Q_{kl})$ and see that the two terms of the first integral cancel when $k=m$ and $l=n$. When $k=n\neq l=m$, we instead find that these terms give the same result when integrated. By using symmetry of the divided difference (i.e., $f'[x,y]=f'[y,x]$) and integrating out all trivial coordinates, we obtain
\begin{align*}
\frac{\int_{H_N} Q_{kl}  Q_{mn} e^{-\frac 12 \br{Q,Q}} dQ}{\int_{H_N} e^{-\frac 12 \br{Q,Q}} dQ} =&  \delta_{kn}\delta_{lm}\frac{2\int_\R (\Re(Q_{kl}))^2 e^{-f'[\lambda_k,\lambda_l](\Re(Q_{kl}))^2}d\Re(Q_{kl})}{\int_\R e^{-f'[\lambda_k,\lambda_l](\Re(Q_{kl}))^2}d\Re(Q_{kl})},
\end{align*}
a Gaussian integral that gives the $G_{kl}$ required by the lemma. When $k=l=n=m$, the result follows similarly.
\end{proof}

The above lemma allows us to leave out all integrals from the subsequent computations. In place of those integrals, we use the following notation.
\begin{defi}\label{def:propagator}
We  define, with slight abuse of notation,
\begin{align*}
	\wick{\c Q_{kl} \hspace{16pt}\c Q_{mn} }:=\delta_{kn}\delta_{lm} G_{kl},
\end{align*}
and refer to $G_{kl}$ as the \textit{propagator}.
\end{defi}

As an example and to fix terminology, we will now compute the amplitudes of the three most basic one-loop diagrams with two external edges.
These are given in Figure \ref{fig:1loop-2pt}.
Using Lemma \ref{lem:propagator} and Definition \ref{def:propagator}, we find the amplitude for the first diagram to be
\begin{align}
  \label{eq:2ptA}
 \hspace{-3pt}\raisebox{-10pt}{\scalebox{0.45}{
\begin{tikzpicture}[thick]
	\draw[edge] (0.7,0) to (1.8,0);
	\draw[edge] (1.8,0) to (2.9,0);
	\draw[edge] (3.1,0) to (4,0);
	\draw[edge] (1.8,0) arc (0:180:-0.6cm);
	\draw[edge] (3,0) arc (180:0:0.6cm);
	\draw[edge] (4,0) to (5.3,0);
	\ncvertex{2,0}
	\ncvertex{4,0}
	\node at (0.3,0) {\huge $V_1$};
	\node at (5.7,0) {\huge $V_2$};
\end{tikzpicture}}}
 &= \sum_{\begin{smallmatrix} i,j,k,l, \\ m,n\leq N\end{smallmatrix}} f'[\lambda_i,\lambda_j,\lambda_k](V_1)_{ij} \wick{ \c1 Q_{jk} \c2 Q_{ki} f'[\lambda_l,\lambda_m,\lambda_n](V_2)_{lm} \c1 Q_{mn} \c2 Q_{nl} }  \nonumber      \nonumber \\
 &  = \sum_{i,k\leq N} f'[\lambda_i, \lambda_i,\lambda_k]f'[\lambda_i, \lambda_k,\lambda_k](V_1)_{ii} (V_2)_{kk} (G_{ik})^2  .
  \end{align}
As $V_1$ and $V_2$ are assumed of finite rank, the above expression converges as $N\to\infty$. To see this explicitly, let $K$ be such that $V_1,V_2\in H_K$, and let $G$ be the diagram on the left-hand side of \eqref{eq:2ptA}. We then obtain
\begin{align}\label{eq:irrelevant diagram}
	\lim_{N\to\infty}\Gamma^G_N(V_1,V_2)=\sum_{i,k\leq K}f'[\lambda_i,\lambda_i,\lambda_k]f'[\lambda_i,\lambda_k,\lambda_k](V_1)_{ii}(V_2)_{kk}(G_{ik})^2,
\end{align}
a finite number. In general we can say that if all summed indices of an amplitude occur in a matrix element of any of the perturbations (e.g., $(V_1)_{ii}$ and $(V_2)_{kk}$) then the amplitude remains finite even when the size $N$ of the random matrices $Q$ is sent to $\infty$. In physics terminology, the first diagram in Figure \ref{fig:1loop-2pt} is \textit{irrelevant}, and can be disregarded for renormalization purposes.

We then turn to the second diagram in Figure \ref{fig:1loop-2pt}, and compute
\begin{align}
\raisebox{-8pt}{\scalebox{0.45}{
\begin{tikzpicture}[thick]
	\draw[edge] (1,0) to (2.1,0);
	\draw[edge] (2.1,0) to[out=90,in=90] (3.9,0);
	\draw[edge] (2.1,0) to[out=-90,in=-90] (3.9,0);
	\draw[edge] (3.9,0) to (5,0);
	\ncvertex{2,0}
	\ncvertex{4,0}
	\node at (0.5,0) {\huge $V_1$};
	\node at (5.5,0) {\huge $V_2$};
\end{tikzpicture}}}
&= \sum_{\begin{smallmatrix} i,j,k,l, \\ m,n\leq N\end{smallmatrix}} f'[\lambda_i,\lambda_j,\lambda_k] (V_1)_{ij} \wick{ \c1 Q_{jk} \c2 Q_{ki} f'[\lambda_l,\lambda_m,\lambda_n](V_2)_{lm} \c2 Q_{mn} \c1 Q_{nl} } \nonumber \\
 &= \sum_{i,j,k\leq N} (f'[\lambda_i, \lambda_j,\lambda_k])^2 (V_1)_{ij} (V_2)_{ji} G_{ik}G_{kj} .
\label{eq:vertexcontr}
\end{align}
This diagram is planar, and the indices $i,j,k$ correspond to regions in the plane, assuming the external edges are regarded to stretch out to infinity. The index $k$ corresponds to the region within the loop, and is called a \textit{running loop index}. As the index $k$ is not restricted by $V_1$ and $V_2$ as in \eqref{eq:2ptA}, we find that in general the amplitude \eqref{eq:vertexcontr} diverges as $N \to \infty$.
In physical terms, this is a \textit{relevant} diagram.

The amplitude of the final diagram becomes
\begin{align}
 \raisebox{-15pt}{\scalebox{0.45}{ 
\begin{tikzpicture}[thick]
	\draw[edge] (0,-1) to (2,0);
	\draw[edge] (4,-1) to (2,0);
	\draw[edge] (1.9,0) arc (-90:270:0.8cm);
	\ncvertex{2,0}
	\node at (-.4,-1) {\huge $V_1$};
	\node at (4.4,-1) {\huge $V_2$};
\end{tikzpicture}}}
\quad&= -\sum_{ i,j,k,l\leq N} f'[\lambda_i,\lambda_j,\lambda_k,\lambda_l](V_1)_{ij} \wick{ \c Q_{jk} \c Q_{kl} } (V_2)_{li}   \nonumber \\
 &= -\sum_{i,j,k\leq N} f'[\lambda_i, \lambda_j,\lambda_j,\lambda_k](V_1)_{ij} (V_2)_{ji} G_{jk} .
 \label{eq:vertexcontr2}
\end{align}
Again, this amplitude contains a running loop index and is therefore potentially divergent in the limit $N \to \infty$. 

\begin{figure}
\hspace{24pt}
    \begin{tabular}{p{.3\linewidth}p{.3\linewidth}p{.3\linewidth}}
 \scalebox{.6}{
    \begin{tikzpicture}[thick]
	\draw[edge] (0.5,0) to (1.8,0);
	\draw[edge] (1.8,0) to (2.9,0);
	\draw[edge] (3.1,0) to (4.2,0);
	\draw[edge] (1.8,0) arc (0:180:-0.6cm);
	\draw[edge] (3,0) arc (180:0:0.6cm);
	\draw[edge] (4.2,0) to (5.5,0);
	\ncvertex{2,0}
	\ncvertex{4,0}
\end{tikzpicture}}
 &\scalebox{.6}{
    \begin{tikzpicture}[thick]
	\draw[edge] (0.5,0) to (2.1,0);
	\draw[edge] (2.1,0) to[out=90,in=90] (3.9,0);
	\draw[edge] (2.1,0) to[out=-90,in=-90] (3.9,0);
	\draw[edge] (3.9,0) to (5.5,0);
	\ncvertex{2,0}
	\ncvertex{4,0}
\end{tikzpicture}}
 &
 \scalebox{.6}{
 \begin{tikzpicture}[thick]
	\draw[edge] (0,-1) to (2,0);
	\draw[edge] (4,-1) to (2,0);
	\draw[edge] (1.9,0) arc (-90:270:0.8cm);
	\ncvertex{2,0}
\end{tikzpicture}}
 \end{tabular}
 
  \caption{Two-point diagrams with one loop. The first one is irrelevant, the second and third are relevant.}
  \label{fig:1loop-2pt}
\end{figure}
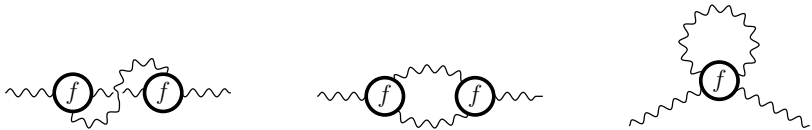

\subsubsection{One-loop counterterms to the spectral action}

Because we are interested in the behavior of the one-loop quantum effective spectral action as $N\to\infty$, we wish to consider only one-loop noncommutative Feynman diagrams whose amplitudes involve a running loop index. For example, the final two diagrams in Figure \ref{fig:1loop-2pt}, but not the first. 

As dictated by the background field method, in order to obtain a quantum effective action we should further restrict to one-particle-irreducible diagrams whose vertices have degree $\geq 3$.

Let us fix a one-loop one-particle-irreducible diagram $G$ in which all vertices have degree $\geq3$, and investigate whether the amplitude of $G$ contains a running loop index. Fix a noncommutative vertex $v$ in $G$. The vertex $v$ will have precisely two incident edges that belong to the loop of the diagram, and at least one external edge. Each index associated with $v$ is associated specifically with two incident edges of $v$. If one of these edges is external, the index will not run, because it will be fixed by the gauge field attached. A running index can only occur if the two incident loop edges of $v$ succeed one another, and the index is placed in between them. The latter of these two loop edges will attach to another noncommutative vertex, $w$, and the possibly running index will also be associated with the succeeding edge in $w$, which also has to be a loop edge if the index is to run. This process may continue throughout the loop until we end up at the original vertex $v$. By this argument, the amplitude of $G$ will contain a running loop index if and only if $G$ can be drawn in the plane with all noncommutative vertices oriented clockwise and all external edges extending outside the loop.

The wonderful conclusion is that the external edges of the relevant diagrams obtain a natural cyclic order. This presents us with a natural one-loop quantization of the bracket $\br{\cdot}$, and thus with a natural proposal for the one-loop quantization of the spectral action.

\begin{defi}\label{def:quantum effective SA}
	Let $N\in\N$ and let $f\in C^\infty$ satisfy $f'[\lambda_i,\lambda_j]>0$ for $i,j\leq N$. We define 
	\begin{align*}
		\bbr{V_1,\ldots,V_n}:=\sum_G \Gamma_N^G(V_1,\ldots,V_n),
	\end{align*}	
	where the sum is over all planar one-loop one-particle-irreducible $n$-point noncommutative Feynman diagrams $G$ with clockwise vertices of degree $\geq3$ and external edges outside the loop and marked cyclically. The \textbf{one-loop quantum effective spectral action} is defined to be the formal series
	$$
\sum_{n=1}^\infty \frac{1}{n} \bbr{V,\ldots,V}.
     $$
\end{defi}
Directly from the definition of $\bbr{\cdot}$, we see that
\begin{align*}
	\bbr{V_2,\ldots,V_n,V_1}=\bbr{V_1,\ldots,V_n}.
\end{align*}
In other words, the property \ref{cyclicity general} holds for the bracket $\brr{\cdot}=\bbr{\cdot}$.
In the next subsection we will show that \ref{commutation general} holds as well.

\subsubsection{Ward identity for the gauge propagator}
In addition to the Ward identity \eqref{eq:ward} for the noncommutative vertex, we claim that we also have the following Ward identity for the gauge edge:
\begin{align}
  \label{eq:ward-gauge}
\raisebox{-20pt}{\scalebox{0.45}{
\begin{tikzpicture}[thick]
	\draw (0,0) arc (-20:20:4cm);
	\draw (-0.025,0) arc (-20:20:4cm);
	\draw (-0.05,0) arc (-20:20:4cm);
	\draw (-0.075,0) arc (-20:20:4cm);
	\draw[edge] (0.15,1.4) to (2.95,1.4);
	\draw (3.1,0) arc (20:-20:-4cm);
	\draw (3.125,0) arc (20:-20:-4cm);
	\draw (3.15,0) arc (20:-20:-4cm);
	\draw (3.175,0) arc (20:-20:-4cm);
	\draw[streepjes] (2.9,1.4) to (2.3,2.9);
	\node at (2.1,3.2) {\huge $a$};
\end{tikzpicture}}}
\quad-\quad
\raisebox{-20pt}{\scalebox{0.45}{
\begin{tikzpicture}[thick]
	\draw (0,0) arc (-20:20:4cm);
	\draw (-0.025,0) arc (-20:20:4cm);
	\draw (-0.05,0) arc (-20:20:4cm);
	\draw (-0.075,0) arc (-20:20:4cm);
	\draw[edge] (0.15,1.4) to (2.95,1.4);
	\draw (3.1,0) arc (20:-20:-4cm);
	\draw (3.125,0) arc (20:-20:-4cm);
	\draw (3.15,0) arc (20:-20:-4cm);
	\draw (3.175,0) arc (20:-20:-4cm);
	\draw[streepjes] (0.2,1.4) to (0.75,2.9);
	\node at (0.9,3.2) {\huge $a$};
\end{tikzpicture}}}
\quad=\quad
\raisebox{-20pt}{\scalebox{0.45}{
\begin{tikzpicture}[thick]
	\draw (0,0) arc (-20:20:4cm);
	\draw (-0.025,0) arc (-20:20:4cm);
	\draw (-0.05,0) arc (-20:20:4cm);
	\draw (-0.075,0) arc (-20:20:4cm);
	\draw[edge] (0.15,1.4) to (1.5,1.4);
	\draw[edge] (1.5,1.4) to (2.95,1.4);
	\draw (3.1,0) arc (20:-20:-4cm);
	\draw (3.125,0) arc (20:-20:-4cm);
	\draw (3.15,0) arc (20:-20:-4cm);
	\draw (3.175,0) arc (20:-20:-4cm);
	\draw[edge] (1.5,1.4) to (1.5,2.8);
	\ncvertex{1.5,1.4}
	\node at (1.5,3.2) {\huge $[D,a]$};
\end{tikzpicture}}}
\end{align}
Indeed, the left-hand side yields terms
  \begin{align*}
    \sum_{m\leq N}\big(\wick{ \c Q_{ik} \c Q_{lm} a_{mn} }-   \wick{ a_{im} \c  Q_{mk}  \c Q_{ln}}\big)
    &=  \sum_{m\leq N}\big(G_{ik} \delta_{im} \delta_{kl} a_{mn} - G_{ln} \delta_{mn} \delta_{kl} a_{im}\big) \\
    &= ( G_{ik}- G_{nk} )\delta_{kl} a_{in},
   \end{align*}
   for arbitrary values of $i$, $k$, $l$, and $n$ determined by the rest of the diagram.
  The right-hand side, by the defining property of the divided difference, and because every internal edge adds a minus sign, yields the terms
   \begin{align*}
  &-\sum_{p,q,r\leq N}\wick{ \c Q_{ik} f'[\lambda_p, \lambda_q, \lambda_r] \c Q_{pq}} [D,a]_{qr}\wick{ \c Q_{rp} \c  Q_{ln}} \\
    &\qquad= - \sum_{p,q,r\leq N}f'[\lambda_p, \lambda_q, \lambda_r](\lambda_q -\lambda_r)a_{qr}G_{ik} \delta_{iq} \delta_{kp} G_{rp} \delta_{rn} \delta_{pl}  \\
    &\qquad=   \left( f'[\lambda_k, \lambda_n] -  f'[\lambda_i, \lambda_k] \right)G_{ik} G_{nk}  \delta_{kl} a_{in}.
    \end{align*}
Because $G_{kl}=1/f'[\lambda_k,\lambda_l]$ (see Lemma \ref{lem:propagator}) the two expressions coincide for every value of $i$, $k$, $l$, and $n$, thereby allowing us to apply the rule \eqref{eq:ward-gauge} whenever it comes up as part of a diagram. For example, by combining \eqref{eq:ward-gauge} with \eqref{eq:ward}, we have

\begin{align*}\label{eq:example quantum Ward}
&\raisebox{-20pt}{\scalebox{.45}{
    \begin{tikzpicture}[thick]
	\draw[edge] (0.5,0) to (2.1,0);
	\draw[edge] (2.1,0) to[out=90,in=90] (3.9,0);
	\draw[edge] (2.1,0) to[out=-90,in=-90] (3.9,0);
	\draw[edge] (3.9,0) to (5.5,0);
	\draw[streepjes] (1.6,0) to (1,-1.5);
	\ncvertex{2,0}
	\ncvertex{4,0}
	\node at (0,0) {\huge $V_1$};
	\node at (6,0) {\huge $V_2$};
	\node at (0.9,-1.7) {\huge $a$};
\end{tikzpicture}}}
\quad
-
\raisebox{-20pt}{\scalebox{.45}{
    \begin{tikzpicture}[thick]
	\draw[edge] (0.5,0) to (2.1,0);
	\draw[edge] (2.1,0) to[out=90,in=90] (3.9,0);
	\draw[edge] (2.1,0) to[out=-90,in=-90] (3.9,0);
	\draw[edge] (3.9,0) to (5.5,0);
	\draw[streepjes] (4.4,0) to (5,-1.5);
	\ncvertex{2,0}
	\ncvertex{4,0}
	\node at (0,0) {\huge $V_1$};
	\node at (6,0) {\huge $V_2$};
	\node at (5.1,-1.7) {\huge $a$};
\end{tikzpicture}}}\\
&\raisebox{-20pt}{
\quad\vspace{20pt}\raisebox{-15pt}{
=
}
\raisebox{-35pt}{\scalebox{.45}{
    \begin{tikzpicture}[thick]
	\draw[edge] (0.5,0) to (2.1,0);
	\draw[edge] (2.1,0) to[out=90,in=90] (3.9,0);
	\draw[edge] (2.1,0) to[out=-90,in=-90] (3.9,0);
	\draw[edge] (3.9,0) to (5.5,0);
	\draw[edge] (2,0) to (1.2,-1.5);
	\ncvertex{2,0}
	\ncvertex{4,0}
	\node at (0,0) {\huge $V_1$};
	\node at (6,0) {\huge $V_2$};
	\node at (1,-2) {\huge $[D,a]$};
\end{tikzpicture}}}
\,\,\raisebox{-15pt}{
+
}\,\,
\raisebox{-45pt}{\scalebox{.45}{
 \begin{tikzpicture}[thick]
	\draw[edge] (-.2,2.7) to (1,2.7);
	\draw[edge] (2,-0.1) to (2,1);
	\draw[edge] (2,1) to (1,2.7);
	\draw[edge] (1,2.7) to (3,2.7);
	\draw[edge] (3,2.7) to (2,1);
	\draw[edge] (3,2.7) to (4.2,2.7);
	\ncvertex{1,2.7}
	\ncvertex{2,1}
	\ncvertex{3,2.7}
	\node at (-0.7,2.7) {\huge $V_1$};
	\node at (4.7,2.7) {\huge $V_2$};
	\node at (2,-0.5) {\huge $[D,a]$};
\end{tikzpicture}}}
\,\,\raisebox{-15pt}{
+
}\,\,
\raisebox{-35pt}{\scalebox{.45}{
    \begin{tikzpicture}[thick]
	\draw[edge] (0.5,0) to (2.1,0);
	\draw[edge] (2.1,0) to[out=90,in=90] (3.9,0);
	\draw[edge] (2.1,0) to[out=-90,in=-90] (3.9,0);
	\draw[edge] (3.9,0) to (5.5,0);
	\draw[edge] (4,0) to (4.8,-1.5);
	\ncvertex{2,0}
	\ncvertex{4,0}
	\node at (0,0) {\huge $V_1$};
	\node at (6,0) {\huge $V_2$};
	\node at (5,-2) {\huge $[D,a].$};
\end{tikzpicture}}}
\!\!
}\nonumber
\end{align*}
The Ward identity for the gauge propagator, in combination with the Ward identity for the fermion propagator \eqref{eq:ward} allows us to derive the so-called {\it quantum Ward identity}:
$$
\bbr{aV_1,\ldots,V_n} - \bbr{V_1,\ldots,V_na} = \bbr{{[D,a],V_1,\ldots, V_n}}.
$$
We derived this identity diagrammatically in \cite{NS21b} for low orders; below we give a general derivation. The quantum Ward identity, in combination with the obvious cyclicity, shows that $\bbr{\cdot}$ is a special case of the generic bracket $\brr{\cdot}$ satisfying property \ref{cyclicity general} and \ref{commutation general} on page \pageref{cyclicity general}, and hence allows us to apply Proposition \ref{prop:bB} and Theorem \ref{thm:asymptotic expansion}. We thus obtain our final result: an expansion of the one-loop quantum effective action in terms of cyclic cocycles.

\begin{figure}
\hspace{.05\linewidth}
 \begin{tabular}{p{.181\linewidth}p{.21\linewidth}p{.19\linewidth}p{.20\linewidth}}
\scalebox{.45}{
\begin{tikzpicture}[thick]
	\draw[edge] (0,-1) to (2,0);
	\draw[edge] (4,-1) to (2,0);
	\draw[edge] (2,-1.6) to (2,0);
	\draw[edge] (1.9,0) arc (-90:270:0.8cm);
	\ncvertex{2,0}
\end{tikzpicture}}
&\scalebox{.45}{
\begin{tikzpicture}[thick]
	\draw[edge] (1,0) to (2,2);
	\draw[edge] (0.5,2) to (2,2);
	\draw[edge] (1,4) to (2,2);
	\draw[edge] (2,2) to[out=90,in=90] (4,2);
	\draw[edge] (2,2) to[out=-90,in=-90] (4,2);
	\draw[edge] (4,2) to (5,4);
	\draw[edge] (4,2) to (5.5,2);
	\draw[edge] (4,2) to (5,0);
	\ncvertex{2,2}
	\ncvertex{4,2}
\end{tikzpicture}}
   	&
 \scalebox{.45}{
\begin{tikzpicture}[thick]
	\draw[edge] (0.5,0.9) to (2,1);
	\draw[edge] (0.5,0) to (2,1);
	\draw[edge] (2,1) to (4,1);
	\draw[edge] (2,1) to (3,2.6);
	\draw[edge] (3,2.6) to (4,1);
	\draw[edge] (3,2.6) to (2.5,4);
	\draw[edge] (3,2.6) to (3.5,4);
	\draw[edge] (4,1) to (5.5,0.9);
	\draw[edge] (4,1) to (5.5,0);
	\ncvertex{2,1}
	\ncvertex{4,1}
	\ncvertex{3,2.6}
\end{tikzpicture}}
  &\quad
 \scalebox{.45}{
\begin{tikzpicture}[thick]
	\draw[edge] (0,1) to (1,2);
	\draw[edge] (0,2) to (1,2);
	\draw[edge] (0,3) to (1,2);
	\draw[edge] (1,2) to (2,1);
	\draw[edge] (2,1) to (3,2);
	\draw[edge] (1,2) to (2,3);
	\draw[edge] (2,3) to (3,2);
	\draw[edge] (2,3) to (1,4);
	\draw[edge] (2,3) to (2,4);
	\draw[edge] (2,3) to (3,4);
	\draw[edge] (3,2) to (4,3);
	\draw[edge] (3,2) to (4,2);
	\draw[edge] (3,2) to (4,1);
	\draw[edge] (2,1) to (1,0);
	\draw[edge] (2,1) to (2,0);
	\draw[edge] (2,1) to (3,0);
	\ncvertex{1,2}
	\ncvertex{2,1}
	\ncvertex{3,2}
	\ncvertex{2,3}
\end{tikzpicture}}
 \end{tabular}
 \caption{Relevant one-loop $n$-point functions with increasing number of vertices.}
 \label{table:skel-1l}
\end{figure}
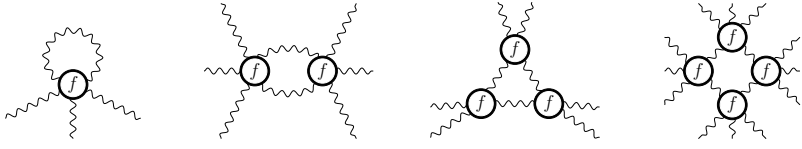

\begin{thm}
There exist $(b,B)$-cocycles $\phi^N$ and $\tilde\psi^N$ (namely, those defined by taking $\brr{\cdot}=\bbr{\cdot}$ in \eqref{eq:def phi_n} and \eqref{eq:psi})
   for which the one-loop quantum effective spectral action can be expanded as
   $$
\sum_{n=1}^\infty \frac{1}{n} \bbr{V,\ldots,V} 
   \sim \sum_{k=1}^\infty \left( \int_{\psi_{2k-1}^N}  \!\!\!\!\! \cs_{2k-1} (A) +\frac 1 {2k} \int_{\phi_{2k}^N} \!\!\!\!\! F^{k} \right).
     $$
   As before, $\tilde\psi_{2k-1}^N=(-1)^{k-1}\tfrac{(k-1)!}{(2k-1)!}\psi_{2k-1}^N$. 
\end{thm}
\begin{proof}
Applying Definition \ref{def:quantum effective SA}, and combining two sums, we obtain
\begin{align*}
	\bbr{aV_1,\ldots,V_n}-\bbr{V_1,\ldots,V_na}=\sum_G \left(\Gamma_N^G(aV_1,\ldots,V_n)-\Gamma_N^G(V_1,\ldots,V_na)\right),
\end{align*}
where the sum is over all \textit{relevant} diagrams $G$, by which we mean the planar one-loop one-particle-irreducible $n$-point noncommutative Feynman diagrams $G$ with clockwise vertices of degree $\geq3$ and external edges outside the loop and marked cyclically.
Let $G$ be a relevant diagram marked $1,\ldots,n$. We let $I(G)$ denote the set of diagrams one can obtain from $G$ by inserting a single gauge edge at any of the places one visits when walking along the outside of the diagram from the external edge $n$ to the external edge $1$. To be precise, if the edges $n$ and $1$ attach to the same noncommutative vertex $v$, we set 
	$$I(G):=\{G'\},$$
where $G'$ is the diagram obtained from $G$ by inserting an external edge marked $n+1$ at $v$ between the edges marked $n$ and $1$. If the edges $n$ and $1$ attach to different vertices $v$ and $w$, respectively, then the edge $e$ succeeding the edge marked $n$ on $v$ necessarily attaches to $w$, preceding the edge marked $1$. In this case, we set
$$I(G):=\{G_n,G_e,G_1\},$$
where $G_n$ is obtained from $G$ by inserting an external edge marked $n+1$ at $v$ between $n$ and $e$, $G_e$ is obtained from $G$ by inserting a noncommutative vertex $v_0$ along $e$ and inserting an external edge marked $n+1$ along the outside of $v_0$, and $G_1$ is obtained from $G$ by inserting an external edge marked $n+1$ at $w$ between $e$ and $1$. By construction of $I(G)$, we find
\begin{align*}
&\llangle a  V_1, \ldots, V_n \rrangle_N^{1L} - 
\llangle V_1, \ldots, V_n  a \rrangle_N^{1L} = \sum_{G}\sum_{G'\in I(G)}
\Gamma_N^{G'}(V_1, \ldots ,V_n,[D,a]).
\end{align*}
The sum over $G$ and $G'$ yields all relevant $n+1$-point diagrams, and, moreover, any relevant $n+1$-point diagram with labels $V_1, \ldots ,V_n,[D,a]$ is obtained in a unique manner from an insertion of an external edge in an $n$-point diagram, as described above. We are therefore left precisely with 
\begin{align*}
\llangle a  V_1, \ldots, V_n \rrangle_N^{1L} - 
\llangle V_1, \ldots, V_n  a \rrangle_N^{1L} = \bbr{V_1,\ldots,V_n,[D,a]}.
\end{align*}
In combination with cyclicity, $\llangle V_1, \ldots, V_n \rrangle_N^{1L} = \llangle V_n , V_1, \ldots, V_{n-1} \rrangle_N^{1L}$, this identity allows us to apply Proposition \ref{prop:bB} and Theorem \ref{thm:asymptotic expansion}. We thus arrive at the conclusion of the theorem.
\end{proof}

We conclude that the passage to the one-loop renormalized spectral action can be realized by a transformation in the space of cyclic cocycles, sending $\phi \mapsto \phi+ \phi^N$ and $\psi \mapsto \psi+ \psi^N$. One could say the theory is therefore one-loop renormalizable in a generalized sense, allowing for infinitely many counterterms, as in \mbox{\cite{GW96}}. Most notably, 
we have stayed within the spectral paradigm of noncommutative geometry.


%

\newcommand{\noopsort}[1]{}\def\cprime{$'$}

\end{document}